\documentclass[authoryear,a4paper, 12pt]{elsarticle}
\usepackage{natbib}
\usepackage{amsthm}
\usepackage{amssymb}
\usepackage{graphicx,amsmath}
\usepackage{tikz}
\usepackage{mathtools}
\usepackage{natbib}
\usepackage{booktabs,xcolor}
\usepackage{graphics}
\usepackage{enumitem}
\usepackage{subcaption}
\usepackage{appendix}
\usepackage{adjustbox}

\usepackage{caption}
\usepackage{setspace}
\usepackage[colorlinks=true,linkcolor=blue, urlcolor=blue, citecolor=blue, breaklinks]{hyperref}
\journal{Economics Letters}

\newtheorem{proposition}{Proposition}

\theoremstyle{definition}

\usepackage{etoolbox}
\makeatletter
\patchcmd{\ps@pprintTitle}
{Preprint submitted to}
{Preprint resubmitted to}
{}{}
\makeatother
\patchcmd{\emailauthor}{(#2)}{}{}{}
\patchcmd{\urlauthor}{(#2)}{}{}{}

\usepackage{tikz}

\newcommand\encircle[1]{%
	\tikz[baseline=(X.base)] 
	\node (X) [draw, shape=circle, inner sep=0] {\strut #1};}

\newcommand\ensquare[1]{%
	\tikz[baseline=(X.base)] 
	\node (X) [draw, shape=rectangle, inner sep=.1cm] {\strut #1};}

\newcommand{\rk}{\mathit{rk}} 
\begin{document}

\begin{frontmatter}
\title{The Losses from Integration in Matching Markets can be Large}
\author{Josu\'e Ortega}
\ead{ortega@zew.de}
\address{Center for European Economic Research (ZEW), Mannheim, Germany.}

\begin{abstract}
	Although the integration of two-sided matching markets using stable mechanisms generates expected gains from integration, I show that there are worst case scenarios in which these are negative. The losses from integration can be large enough that the average rank of an agent's spouse decreases by 37.5\% of the length of their preference list in any stable matching mechanism.
\end{abstract}

\end{frontmatter}
	
	\newpage
	\setcounter{footnote}{0}

\section{Introduction}
\label{sec:introduction}

The theory of two-sided matching has been helpful to analyze and improve the allocation of goods without money, such as the assignment of students to public schools or of kidneys available for transplantation to patients in need. In particular, a recent branch of the literature has focused on understanding why small matching markets struggle to integrate into large centralized clearinghouses. For example, \cite{ashlagi2014} and \citet{agarwal2018} document that American hospitals refrain to disclose their available kidneys for transplantation to centralized clearinghouses and even refrain to join them at all (62\% of kidney exchange transplants are within hospital transplants that are not facilitated by a centralized clearinghouse). Similarly, \cite{ekmekci2018} show that schools have incentives to conduct their own independent admissions instead of joining a centralized admission system. 

In both cases, the failure of integration has important consequences: if all hospitals were fully integrated into a centralized kidney exchange program, they would be able to conduct more transplants and thus would save more lives. Similarly, the integration of school admissions into a unified process would reduce both the number of vacant seats in schools and the number of unassigned students, while at the same time reducing the severe racial segregation within school districts \citep{hafalir2017}.

The aforementioned articles document several reasons of why matching markets fail to integrate when doing so could generate welfare gains (and estimations suggest that in some markets the gains from integration are substantial, as in the case of kidney exchange in US \citep{agarwal2018}). In this paper, I show theoretically that, in worst-case scenarios, matching markets may fail to integrate because doing so would inevitable generate welfare losses in a well-defined sense. This worst-case result is somewhat surprising since the integration of two-sided matching markets never harms more people than those it benefits ex-post and generates expected welfare gains for both sides of the market \citep{ortega2018}.

In my framework, there are $\kappa$ disjoint Gale-Shapley matching problems, each of size $n$. Agents are first assigned a spouse within their community so that the matching in each market is stable. Then all communities integrate as a society and a stable matching is realized again, now allowing agents to marry any agent of opposite gender in the society. I measure the individual gains from integration as the difference in the spouse ranking before and after integration occurs. This is, if an agent obtains his 5th best partner before integration and his 2nd best after integration, there is a difference in ranking of +3 and this agent experiences gains from integration. The spouse ranking is a common welfare measure in matching markets \citep{wilson1972, pittel1989,ashlagi2017}. 

Similarly, the total gains from integration of a society are the sum of the gains from integration of each of its agents. The sum of spouse rankings is a natural welfare measure previously studied by \cite{knuth1976,knuth1996} and \cite{knoblauch2009} to analyze the properties of random markets. Thus, the sum of individual gains from integration is an intuitive way to measure the degree to which integration is perceived as favourable.

I show with an example that in specific matching markets the gains from integration can be negative. By providing an algorithm on how to create replicas of a matching problem, I show that the losses from integration can be made large enough that on average the ranking of each agent's spouse decreases by $\frac{3}{8}$ of the length of their preference list, equivalent to a welfare reduction of 37.5\%.

\section{Model}
\label{sec:model}

A {\it community} $C$ is a set of $n$ men and $n$ women. There are $\kappa$ disjoint communities. Given a set of communities, the {\it society} $S$ is the set of all communities. A {\it population} $P \subseteq S$ is a set of communities in $S$. $M^P$ and $W^P$ denote the sets of men and women in each community $C \in P$. $C_i$ denotes the $i$-th community and $m_j^i$ denotes the $j$-th man of community $C_i$. I use the same notation for women $w_j^i$ and omit indexes whenever no confusion arises.

Each man $m$ (resp. woman $w$) has strict preferences over the set of all women in the society $W^S$ (resp. men $M^S$). I assume for exposition that all potential partners are desirable. I write $w \succ_m  w'$ to denote that $m$ prefers $w$ to $w'$. Given a population $P$, I call $\succ_P \coloneqq (\succ_x)_{x \in (M^P \cup W^P)}$ a {\it preference profile} of the agents belonging to the communities in $P$. An {\it extended matching problem} (EMP) is a triple $(M^S,W^S;\succ_S)$. 

Given a population $P$, a {\it matching} $\mu:(W^P \cup M^P) \mapsto (W^P \cup M^P$) is a function of order two ($\mu^2(x)=x)$ such that if $\mu(m) \neq m$, then $\mu(m) \in W^P$, and if $\mu(w) \neq w$, then $w \in M^P$. An agent is matched to himself if he remains unmatched. A {\it matching scheme} $\sigma: (M^S \cup W^S) \times 2^S \mapsto (M^S \cup W^S)$ is a function that specifies a matching $\sigma(\cdot,P)$ for every $P \in 2^S$, i.e. $\sigma (\cdot,P): (M^P \cup W^P) \mapsto (M^P \cup W^P)$. Here we will only be interested in $\sigma(\cdot, C)$ and $\sigma(\cdot, S)$ which denote the matching before and after integration occurs. Matching schemes are analogous to the concept of assignment schemes in cooperative game theory \citep{sprumont1990}.

A matching $\mu:(W^P \cup M^P) \mapsto (W^P \cup M^P)$ is stable if there is no man $m \in M^P$ and no woman $w \in W^P$ that are not married to each other $(\mu(m) \neq w)$ such that $w  \succ_m \mu(m)$ and $m \succ_w \mu(w)$. Any such pair ($m,w$) is called a blocking pair. A matching scheme $\sigma$ is stable if the matching $\sigma(\cdot,P)$ is stable in every population $P \subseteq S$. 

The {\it raw rank} of a woman $w$ in the preference order of a man $m$ over all potential spouses in the society is defined by $\rk_m(w) \coloneqq \left\vert \{ w' \in W^S: w' \succcurlyeq_m w \} \right\vert$. Similarly, $\rk_w(m)$ denotes the raw rank of $m$ in the preference order of $w$. Note that the length of each preference order is $\kappa n$, i.e. each agent ranks $\kappa n$ different potential partners. As it will be more convenient, we will use the \emph{percentile ranking} instead of the raw ranking, which is defined as $\hat \rk_m(w) \coloneqq \frac{\rk_m(w)}{\kappa n}$. Thus, if $\rk_m(w)=\kappa n$, then $\hat \rk_m(w)=1$, meaning that the fraction of potential partners who are equally or more preferred than $w$ by $m$ is $1$, i.e. everybody.

The {\it gains from integration} for agent $x$ under the matching scheme $\sigma$ are defined as $\gamma_x (\sigma) \coloneqq \hat \rk_M(\sigma(x,C)) - \hat \rk_M(\sigma(x,S))$. The {\it total gains from integration} are given by $\Gamma (\sigma) \coloneqq \sum_{x \in S} \gamma_x$. If these are negative, we speak of the total losses from integration. The average gains from integration are denoted by $\overline \Gamma (\sigma)=\frac{\Gamma (\sigma)}{2 \kappa n}$. Note that $\overline \Gamma (\sigma) \in (-1,1)$.

\section{Results}
\label{sec:results}

Unfortunately, integration produces total welfare losses in some EMPs if the matching before and after integration occurs is stable.

\begin{proposition}
	\label{thm:negative}
	There exist EMPs such that $\overline \Gamma(\sigma)<0$ in any stable matching scheme. 
\end{proposition}

\begin{proof} Let $n=2$ and $\kappa=2$. Consider the following preferences (preferences that remain unspecified are irrelevant; the partner obtained before integration appears in a circle, whereas the partner obtained after integration appears in a square):
\begin{center}
	$ \begin{matrix*}[l]
m_1^1: \ensquare{$w_1^2$} \succ \encircle{$w_2^1$}  	 &\quad w_1^1: \ensquare{$m_1^2$} \succ \encircle{$m_2^1$}  \\
m_2^1: \encircle{$w_1^1$} \succ w_1^2 \succ \ensquare{$w_2^1$}   &\quad w_2^1: \encircle{$m_1^1$} \succ m_1^2 \succ \ensquare{$m_2^1$}   \\ \\

m_1^2: \ensquare{$w_1^1$} \succ \encircle{$w_2^2$}   &\quad w_1^2: \ensquare{$m_1^1$} \succ \encircle{$m_2^2$}   \\
m_2^2: \encircle{$w_1^2$} \succ w_1^1 \succ w_2^1 \succ \ensquare{$w_2^2$} &\quad w_2^2: \encircle{$m_1^2$} \succ m_1^1 \succ m_2^1 \succ \ensquare{$m_2^2$}
	\end{matrix*}$
\end{center}

In any stable matching  $\sigma(m_1^i, C_i)=w_2^i$ and $\sigma(m_2^i, C_i)=w_1^i$ for $i \in \{1,2\}$ because all agents get their most preferred partner inside their own community and thus the matching $\sigma(\cdot, C_i)$ is stable. In any stable matching $\sigma(m_1^i, S)=w_1^j$ for $i \in \{1,2\}, i \neq j$ because $m_1^i$ and $w_1^j$ are the best possible spouses for each other, and thus none of these agents could be in any blocking pair. Then $\sigma(m_2^1,S)=w_2^1$ because both get their best achievable partner among those remaining. Finally, $\sigma(m_2^2,S)=w_2^2$ so that both agents obtain their worst partner in $S$, yet no other partner prefers them over their current partner, and thus the matching $\sigma(\cdot, S)$ is stable. We conclude that the matching scheme $\sigma$ is stable.

The gains from integration (in raw ranking terms) for men in this unique stable matching scheme are $-3$ rank positions, and by symmetry of the problem the gains from integration (in raw ranking terms) for women are $-3$ as well. Normalizing them in percentile terms, we have that $\Gamma_S(\sigma)=-\frac{6}{4}$ and $\overline \Gamma_S(\sigma)=-\frac{3}{16}$.
\end{proof}

What drives the existence of welfare losses from integration in the previous example is that agents' preferences are constructed so that those who benefit from integration (exactly half of the society) experience very small welfare gains (+1 in a raw rank measure). In contrast, those harmed by integration (the other half of the society) find themselves married to much less preferred partners. Neither the fact that exactly half of the society experiences losses from integration nor the particular way in which preferences were constructed are a coincidence, and we will return to this example in Proposition \ref{thm:lowerbound}.

A natural follow-up question is how big can total losses from integration be. A trivial upper bound is easy to compute.

\begin{proposition}
	\label{thm:maximum}
	$\overline \Gamma(\sigma) \geq -  \frac{1}{2}+\frac{1}{\kappa n}$ in any EMP and any stable matching scheme.
\end{proposition}

\begin{proof}
	From Proposition 2 in \cite{ortega2018}, we know that the number of agents who get hurt by integration is weakly lower than the number of agents who benefit from integration. The smallest gain from integration is (+1) in raw ranking terms. The largest loss from integration is $(\kappa n-1)$ in raw ranking terms. Thus
	\begin{equation}
		\label{eq:maximum}
		\Gamma (\sigma) \geq \kappa n(\frac{-\kappa n+1}{\kappa n})+ \kappa n (\frac{1}{\kappa n})= - \kappa n+2
	\end{equation} 
	
	and $\overline \Gamma (\sigma) \geq - \frac{1}{2} +\frac{1}{\kappa n}$. The term $\frac{1}{\kappa n}$ becomes negligible in large markets and thus the inequality asymptotically becomes $\overline \Gamma (\sigma) \geq - \frac{1}{2}$.
\end{proof}

Expression \ref{eq:maximum} establishes that the total losses from integration could be such that the ranking of agents' spouses goes down by almost 50\% of the length of their preference order, a very big loss which becomes even larger as the size of the market grows. However, the above upper bound is not achievable, because in the case of $n=2$ and $\kappa=2$ the greatest lower bound for $\Gamma(\sigma)$ is $-\frac{3}{16}$ instead of $-\frac{1}{4}$. In the following Proposition, I provide an attainable bound on the losses from integration, which shows that the losses of integration can be substantial and do not vanish in large matching markets. 

\begin{proposition}
	\label{thm:lowerbound}
	For any $n$ and $k$ even, there exists an EMP where $\overline \Gamma(\sigma) \leq -\frac{3}{8}  + \frac{3}{4\kappa n}$ in any stable matching scheme.
\end{proposition}
\begin{proof}
I show that one can generate such an EMP by cloning the society in the proof of Proposition \ref{thm:negative}. Cloning a matching problem was first done by \cite{irving1986} to establish an upper bound on the number of stable marriages. The approach of creating clones or replicas of an economy has been used to establish other properties of matchings in large markets \citep{wooders1998,kojima2010,azevedo2015,liu2016,he2018} and has a long tradition in economics since its use by \cite{debreu1963} to prove the equivalence of the core and the competitive equilibrium in large markets. Clones are denoted by a prime symbol (i.e. the clone of agent $m_i^j$ is denoted by $m_{i'}^j$) and the set of clone men and women are denoted by $M'$ and $W'$, respectively.

I use the notation $M^+$ and $W^+$ to denote the men and women who benefit from integration. Similarly, the notation $M^-$ and $W^-$ denotes the men and women harmed by integration. The replication of the EMP will be such that the clone of each agent in $M^+ \cup  W^+$ will also be in $M'^+ \cup W'^+$, and the clone of every agent in $M^- \cup W^-$ will be in $M'^- \cup W'^-$. 

The preferences of the clones in $M'^+ \cup W'^+$ will be isomorphic to those of the original agents in $M^+ \cup W^+$, meaning that if an agent $m_i^j$ prefers $w_l^h$ over $w_f^p$, then agent $m_{i'}^j$ prefers $w_{l'}^h$ over $w_{f'}^p$. The preferences of the clones in $M'^+ \cup W'^+$ will be such that
\begin{enumerate}
	\item Their most preferred partner will be isomorphic to the one of the original agent.
	\item Women will prefer any agents in $M^+ \cup M'^+$ over any agent in $M^- \cup M'^-$ (the same is true for men). 
	\item The preferences of women (resp. men) over $M^- \cup M'^-$ (resp. $W^- \cup W'^-$) are lexicographic according to the following criteria:
    	\begin{enumerate}
	    	\item Agent $x^i_j$ is preferred to agent $x^h_p$ if $i < j$ (agents from communities with smaller indexes are preferred),
	    	\item If both agents belong to the same community, agent $x_i^h$ is preferred to agent $x_j^h$ if $i<j$, 
	    	\item If both agents have the same sub- and superindices, then the original is preferred to the clone.
    	\end{enumerate}
\end{enumerate}

The EMP below shows the corresponding replica of the EMP in the proof of Proposition \ref{thm:negative}. 
\begin{center}
	$ \begin{matrix*}[l]
	m_1^1: \ensquare{$w_1^2$} \succ \encircle{$w_2^1$}  	 &\quad w_1^1: \ensquare{$m_1^2$} \succ \encircle{$m_2^1$}  \\
	m_{1'}^1: \ensquare{$w_{1'}^2$} \succ \encircle{$ w_{2'}^1$}  	 &\quad w_{1'}^1: \ensquare{$m_{1'}^2$} \succ \encircle{$m_{2'}^1$} \\
	m_2^1: \encircle{$w_1^1$} \succ W^+ \succ \ensquare{$w_2^1$}   &\quad w_2^1: \encircle{$m_1^1$} \succ M^+ \succ \ensquare{$m_2^1$}   \\
	m_{2'}^1: \encircle{$w_{1'}^1$} \succ W^+ \succ w_2^1 \succ \ensquare{$w_{2'}^1$}   &\quad w_{2'}^1: \encircle{$m_1^1$} \succ M^+ \succ m_1^2 \succ \ensquare{$m_{2'}^1$}   \\ \\
	
	m_1^2: \ensquare{$w_1^1$} \succ \encircle{$w_2^2$}   &\quad w_1^2: \ensquare{$m_1^1$} \succ \encircle{$m_2^2$}   \\
	m_{1'}^2: \ensquare{$w_{1'}^1$} \succ \encircle{$w_{2'}^2$}   &\quad w_{1'}^2: \ensquare{$m_{2'}^1$} \succ \encircle{$m_{1'}^2$}   \\
	m_2^2: \encircle{$w_1^2$} \succ W^+ \succ w_2^1 \succ w_{2'}^1 \succ \ensquare{$w_2^2$} &\quad w_2^2: \encircle{$m_1^2$} \succ M^+ \succ m_1^2 \succ  m_{2'}^1 \succ \ensquare{$m_2^2$}\\
	m_{2'}^2: \encircle{$w_{1'}^2$} \succ W^+ \succ w_2^1 \succ w_{2'}^1 \succ w_2^2 \succ \ensquare{$w_{2'}^2$} &\quad w_{2'}^2: \encircle{$m_{1'}^2$} \succ M^+ \succ m_1^2 \succ  m_{2'}^1 \succ m_2^2 \succ \ensquare{$m_{2'}^2$}
	\end{matrix*}$
\end{center}

An identical argument to the one we used in the proof of Proposition \ref{thm:negative} shows that the EMP above has a unique stable matching scheme, presented above using circles and squares. In such unique stable matching scheme, $\kappa n$ agents gain (+1) from integration in raw ranking terms, and $\kappa n$ agents have welfare losses ranging from $\kappa n / 2$ to $\kappa n - 1$ also in raw ranking terms.

The total gains from integration are therefore given by
\begin{eqnarray*}
\Gamma (\sigma) &=& \kappa n \left( \frac{1}{\kappa n}\right) - \frac{2}{\kappa n} [ \sum_{i=0}^{\frac{\kappa n}{2} -1} \frac{\kappa n}{2}  + i]\\
&=& 1 -  \frac{2}{\kappa n} [ \frac{(\kappa n)^2}{4}  + \sum_{i=1}^{\frac{\kappa n}{2} -1}  i]\\
&=& 1 -  [ \frac{\kappa n}{2}  + \frac{\kappa n}{4} - \frac{1}{2}]\\
&=& \frac{3}{2}  - \frac{3}{4} \kappa n
\end{eqnarray*}

And thus $\overline{\Gamma} (\sigma) = - \frac{3}{8} + \frac{3}{4 \kappa n}$. The constant term $\frac{3}{4 \kappa n}$ becomes negligible in large markets and thus the total losses of integration can be made asymptotically equivalent to $- \frac{3}{8}$, i.e. a $37.5\%$ of the size of the agents' preference lists.
\end{proof}

\section{Conclusion}

We can construct matching problems such that the sum of the individual gains from integration from all its agents is negative. We can make the problem arbitrarily large to show that the welfare losses remain significant even in large markets.

In this article I have assumed that all communities have the same number of men and women, an unrealistic assumption which I have introduced to simplify the analysis. Allowing communities to have different gender composition would reduce the magnitude of welfare losses. To observe this, suppose a male-dominated community merges with a female-dominated one. In both communities some agents, who were unmatched before integration, are now able to find a partner in a more gender-balanced society. Those agents would experience welfare gains which are not covered in my analysis. I leave a more comprehensive study of integration of unbalanced matching markets for further research.

\section*{Acknowledgements}

I acknowledge excellent research assistance by Alexander Sauer and Carlo Stern, and helpful comments from Yan Long, Gabriel Ziegler and an anonymous referee. Sarah Fox proofread the paper.

\section*{References}
\bibliographystyle{ecta}

\end{document}